\documentclass[11pt]{article}

\usepackage[utf8]{inputenc}
\usepackage[T1]{fontenc}
\usepackage{lmodern}

\usepackage{makeidx}
\usepackage{graphicx,comment}
\usepackage[linesnumbered,vlined]{algorithm2e}

\usepackage{amsmath}
\usepackage{amsthm}
\usepackage{amsfonts, amssymb}
\usepackage{fullpage}
\usepackage[noadjust]{cite}
\usepackage{paralist}

\newcommand{\etal}{\textit{et al.}\xspace}
\newenvironment{lemma-repeat}[1]{\begin{trivlist}
		\item[\hspace{\labelsep}{\bf\noindent Lemma \ref{#1} }]\em }%
	{\end{trivlist}}
\newenvironment{theorem-repeat}[1]{\begin{trivlist}
		\item[\hspace{\labelsep}{\bf\noindent Theorem \ref{#1} }]\em }%
	{\end{trivlist}}

\newcommand*\samethanks[1][\value{footnote}]{\footnotemark[#1]}

\newcommand{\Deal}{request}
\newcommand{\Vault}{vault}
\newcommand{\Bank}{bank}
\newcommand{\Budget}{budget}
\newcommand{\InCover}{\texttt{InCover}}
\newcommand{\NotInCover}{\texttt{NotInCover}}
\newcommand{\ceil}[1]{\left\lceil #1 \right\rceil}
\newcommand{\floor}[1]{\left\lfloor #1 \right\rfloor}

\newcommand{\parentheses}[1]{\left(#1\right)}

\newcommand{\set}[1]{\left\{#1\right\}}

\newcommand{\Dv}{\ensuremath{\Delta}}
\newcommand{\Kv}{\ensuremath{K}}
\newcommand{\eps}{\epsilon}

\makeatletter
\newcommand\footnoteref[1]{\protected@xdef\@thefnmark{\ref{#1}}\@footnotemark}
\makeatother
\newtheorem{theorem}{Theorem}
\newtheorem{claim}{Claim}
\newtheorem{lemma}{Lemma}

\begin{document}

\begin{titlepage}
	%\title{A Distributed $(2+\epsilon)$-Approximation for Minimum Weight\\ Vertex Cover in $O(\log{\Delta}/(\log\log{\Delta}+\log{\epsilon}))$ Rounds}
	\title{A Deterministic Distributed $2$-Approximation for Weighted Vertex Cover in $O(\log n\log\Delta / \log^2\log\Delta)$ Rounds}
	\author{Ran Ben-Basat\thanks{Technion, Department of Computer Science, \texttt{sran@cs.technion.ac.il}. This work was partially sponsored by the Technion-HPI reseach school.}
		\and Guy Even \thanks{Tel Aviv University, \texttt{guy@eng.tau.ac.il}.} \and Ken-ichi Kawarabayashi \thanks{NII, Japan, \texttt{k\_keniti@nii.ac.jp}, \texttt{greg@nii.ac.jp }. This work was supported by JST ERATO Grant Number JPMJER1201, Japan.} \and Gregory Schwartzman\samethanks}
%	\author{Reuven Bar-Yehuda\thanks{Technion, Department of Computer Science, \texttt{reuven@cs.technion.ac.il}, \texttt{ckeren@cs.technion.ac.il}, \texttt{gregorys@cs.technion.ac.il}. Supported in part by the Israel Science Foundation (grant 1696/14).}
%		\and Keren Censor-Hillel\samethanks \and Gregory Schwartzman\samethanks}
	\maketitle

        \begin{abstract}
          We present a deterministic distributed $2$-approximation
          algorithm for the Minimum Weight Vertex Cover problem in the
          CONGEST model whose round complexity is $O(\log n\log\Delta
          / \log^2\log\Delta)$. This improves over the currently best
          known deterministic 2-approximation implied by
          \cite{KhullerVY94}. 
          Our solution generalizes the $(2+\epsilon)$-approximation algorithm of \cite{Bar-YehudaCS17}, improving the dependency on $\epsilon^{-1}$ from linear to logarithmic.  
%          We present a generalization of the $(2+\epsilon)$-approximation algorithm of \cite{Bar-YehudaCS17}, which improves the dependence on $\epsilon^{-1}$ from linear to logarithmic.  
          In addition, for every $\eps=(\log \Delta)^{-c}$, where $c\geq 1$ is a constant, our algorithm computes a $\parentheses{2+\eps}$-approximation
          in $O(\log{\Delta}/\log\log{\Delta})$~rounds (which is
          asymptotically optimal).
	\end{abstract}

	\thispagestyle{empty}
\end{titlepage}
%\maketitle

%	\thispagestyle{empty}
%\end{titlepage}

\section{Introduction}
The Minimum Weight Vertex Cover Problem (MWVC) is defined as
follows. The input is a graph $G=(V,E)$ with nonnegative vertex
weights $w(v)$. A subset $U\subseteq V$ is a \emph{vertex cover} if,
for every edge $e=\{u,v\}$, the intersection $U\cap \{u,v\}$ is not
empty. The weight of a subset of vertices $U$ is $\sum_{v\in
  U}w(v)$. The goal is to find a minimum weight vertex cover. This
problem is one of the classical NP-hard problems~\cite{Karp72}. 

In this paper we deal with distributed deterministic approximation
algorithms for MWVC. We focus on the CONGEST model of distributed
computation in which the communication network is the graph $G$
itself.\footnote{In the CONGEST model vertices have distinct IDs (that
  are polynomial in $|V|$), however, as in~\cite{Bar-YehudaCS17}, our
  algorithm works also in the case of anonymous vertices.}
Computation proceeds in synchronous rounds. Each round consists of
three parts: each vertex receives messages from its neighbors,
performs a local computation, and sends messages to its neighbors. The sent messages
arrive at their destination in the beginning of the next round. In the CONGEST model,
message lengths are bounded by $O(\log |V|)$.  In order to send 
vertex weights, we assume that all the vertex weights are
positive integers bounded by polynomial in $n\triangleq |V|$.  See~
\cite{Bar-YehudaCS17,AstrandS10} for detailed overviews of distributed
algorithms for MWVC.

Let $\Delta$ denote the maximum vertex degree in $G$.  Two of the most
relevant results in this setting to our paper are the lower bound
of~\cite{KuhnMW16} and the upper bound of~\cite{Bar-YehudaCS17}.  The
lower bound of Kuhn \etal~\cite{KuhnMW16} states that every constant
approximation algorithm for MWVC requires at least $\Omega(\log \Delta/\log
\log \Delta)$ rounds of communication. The upper bound of Bar-Yehuda
\etal~\cite{Bar-YehudaCS17} presents a deterministic distributed $(2+\eps)$-approximation 
algorithm (BCS Algorithm) that requires $O(\log \Delta/(\eps \cdot \log \log \Delta))$
rounds for $\eps\in(0,1)$. For
$\epsilon =\Omega(\log \log \Delta / \log \Delta)$, the running time
is $O(\log{\Delta}/\log\log{\Delta})$, with no dependence on
$\epsilon$, and is optimal according to \cite{KuhnMW16}.

In this paper, we present a generalization of the BCS Algorithm with
improved guarantees on the running time for certain ranges of
$\epsilon$.  We focus on decreasing the dependency of the number of
rounds on $\eps$. Since the round complexity of the BCS Algorithm is
optimal for constant values of $\eps$ (and even $\eps=\Omega(\log
\log \Delta / \log \Delta)$), we consider values of $\eps$ that depend on $\Delta$. 

Our main result\footnote{All logarithms
  are base $2$ unless the basis is written explicitly.} is a deterministic distributed
$(2+\eps)$-approximation algorithm in which the number of rounds is
bounded by
\begin{align*}
  O\parentheses{\frac{\log \Dv}{\log\log \Dv} +
  {\frac{\log\epsilon^{-1}\log \Dv}{\log^2 \log \Dv}}}.
\end{align*}
This result assumes that all the vertices know $\Delta$ or an estimate
that is a polynomial in $\Delta$. This result leads to the following consequences:
\begin{enumerate}
\item If $\eps^{-1}=(\log \Delta)^c$, for a constant $c>0$, then the
  number of rounds asymptotically matches the lower bound, and is thus
  optimal.  In~\cite{Bar-YehudaCS17} the same asymptotic running time is
  guaranteed only for $\eps^{-1}=O(\log\Delta/\log \log \Delta)$.
\item If $\eps^{-1}= (\log \Delta)^{\omega(1)}$, then
  the dependency of the round complexity on $1/\eps$ is reduced from
  linear to logarithmic. In addition, the round complexity is decreased by an additional factor of $\log\log \Delta$.
\item Every $(2+\eps)$-approximation is a $2$-approximation if
  $\eps<1/(nW)$, where $W=\max_v w(v)$. Since we assume that
  $W=n^{O(1)}$, where $n=|V|$, we obtain a $2$-approximate
  deterministic distributed algorithm for MWVC with round complexity
  $O(\log n \cdot \log \Delta/ \log^2 \log \Delta)$.  This improves over the
  2-approximation in $O(\log^2 n)$ rounds implied by
  \cite{KhullerVY94}\footnote{The actual
    result is stated as a $(2+\epsilon)$-approximation in $O(\log
    \epsilon^{-1} \log n)$ rounds, from which we infer a
    2-approximation by setting $\epsilon=1/nW$.} (which has the lowest round complexity for
  deterministic $2$-approximation to the best of our knowledge).
\end{enumerate}

Our round complexity increases for the case that the maximum degree
$\Delta$ is unknown to the vertices of the graphs. We propose two
alternatives for the case that $\Delta$ is unknown.  The first alternative holds for every
$\eps\in(0,1)$, and achieves a $(2+\eps)$-approximate solution  with a
round complexity of
$O\parentheses{\frac{\log\epsilon^{-1}\log \Dv}{\log \log \Dv}}$.
The second alternative holds for $\eps>(\log \Delta)^q$, where $q>0$ is a constant.
In the second alternative, a $(2+\eps)$-approximation is achieved with an optimal asymptotic round complexity of
$O(\log \Delta/\log \log \Delta)$.

Our algorithm builds on the BCS Algorithm~\cite{Bar-YehudaCS17}.  This
algorithm adapts the \emph{local ratio}
framework~\cite{BarYehudaE1985} to the distributed setting, with
several improvements that provide the desired speedup. The BCS
Algorithm can be also interpreted as the following ``primal-dual'' 
algorithm. Essentially the algorithm aims to increase the edge
variables (i.e., {\it dual}) such that the following holds:
\begin{compactenum}
	\item  The sum of edge variables incident to every node does not exceed its weight ({\it feasibility of edge variables}).
	\item The set of vertices whose edge variable sum is at least
          $(1-\epsilon)$-fraction of the vertex weight constitute a
          {\it vertex cover}.
\end{compactenum}
The above conditions yield a $(2+\epsilon)$-approximation for MWVC.

The challenge in the above framework is to maintain feasibility of the
edge variables while converging as fast as possible to a vertex
cover. To increase the edge variables, vertices send \emph{offers} to
their neighbors. The neighbors respond to these offers in a way that
guarantees feasibility of the edge variables.  This requires a
coordination mechanism in the distributed setting, as a vertex both
sends and receives offers simultaneously. To this end, the weight of
every vertex is divided into two parts: \emph{vault} and
\emph{bank}. Offers are allocated from the vault, while responses are
allocated from the bank, respectively.  Hence the agreed upon increases to the edge
variables do not violate the feasibility of the edge
variables.  The BCS algorithm sets the vault to be an
$\epsilon$-fraction of the vertex weight (and the bank to be the
remainder). This leads to a running time of $O(\epsilon^{-1} \log
\Delta / \log \log \Delta)$ and $O(\log \Delta / \log \log \Delta)$ if
$\epsilon = \Omega( \log \log \Delta / \log \Delta)$.

Our algorithm introduces three modifications to the BCS Algorithm, which allows us to 
improve the round complexity.  First, we attach \emph{levels} to the vertices
that measure by how much the remaining weight of a vertex has
decreased.  Second, the size of the vault decreases as the level of
the vertex increases.  Third, offers are not sent to all the
neighbors. Instead,  offers are sent only to the neighbors whose level
is the smallest level among the remaining neighbors.  

\paragraph{Related Work}
An excellent overview of the related work is presented in
\cite{Bar-YehudaCS17,AstrandS10} which we summarize hereinafter.
Minimum vertex cover is one of Karp's 21 NP-hard
problems~\cite{Karp72}.  A simple 2-approximation for the unweighted
version can be achieved by a reduction from maximal matching (see,
e.g.,~\cite{Cormen2009,GareyJ79}).  For the weighted case,
\cite{BarYehudaE81} achieves the first linear-time $2$-approximation
algorithm using the primal-dual schema, while \cite{BarYehudaE1985}
achieves the same result using the local-ratio technique. Prior to
that, the first polynomial-time $2$-approximation algorithm was due
to~\cite{NemhauserT75} and observed by~\cite{Hochbaum82}. For any
constant $\epsilon>0$, if the Unique Games conjecture holds, no
polynomial-time algorithm can compute a $(2-\epsilon)$ approximation
of the minimum vertex cover~\cite{KhotR08}.

Let us now turn our attention to the distributed setting. Let us start
from the unweighted case. A $2$-approximation can be found in
$O(\log^4{n})$ rounds~\cite{HanckowiakKP01} and in
$O(\Delta+\log^{*}{n})$ rounds~\cite{PanconesiR01}. Completely local
algorithms with no dependence on $n$ are presented
in~\cite{AstrandFPRSU09} which gives an $O(\Delta^2)$-round
$2$-approximation algorithm, and in~\cite{PolishchukS09} which gives
an $O(\Delta)$-round $3$-approximation algorithm. Using the maximal
matching algorithm of~\cite{BarenboimEPS12} gives a $2$-approximation
algorithm for vertex cover in $O(\log{\Delta}+(\log\log{n})^4)$
rounds. This can be made into a
$(2+1/\text{poly}{\Delta})$-approximation algorithm within
$O(\log{\Delta})$ rounds~\cite{PettiePersonal}.

For the weighted case, \cite{GrandoniKP08} presents a randomized
$2$-approximation algorithm in $O(\log{n}+\log{W})$ rounds (where $W$
is a bound on the vertex weights). In~\cite{KoufogiannakisY09} the
first (randomized) $2$-approximation algorithm running in $O(\log{n})$
rounds is presented (note that the running time is logarithmic in $n$
and independent of the weights). A deterministic $2$-approximation
algorithm in $O(\Delta+\log^{*}{n})$ rounds is given
within~\cite{PanconesiR01}. In~\cite{KhullerVY94}, a deterministic
$(2+\epsilon)$-approximation algorithm is given within
$O(\log{\epsilon^{-1}}\log{n})$ rounds. As for deterministic
algorithms independent of $n$,~\cite{KuhnMW06} presents a
$(2+\epsilon)$-approximation algorithm in
$O(\epsilon^{-4}\log{\Delta})$ rounds and~\cite{AstrandFPRSU09}
presents a $2$-approximation algorithm in $O(1)$ rounds for
$\Delta \leq 3$, while~\cite{AstrandS10} presents a $2$-approximation
algorithm in $O(\Delta+\log^* W)$ rounds (where
$W\triangleq \max_v w(v)$). Finally in \cite{Bar-YehudaCS17} a
deterministic $(2+\epsilon)$-approximation which runs in
$O(\epsilon \log \Delta / \log \log \Delta)$ rounds is given. In
\cite{Solomon18} a $(2+\epsilon)$-approximation in
$O(\epsilon^{-1} \log (\alpha / \epsilon) / \log \log (\alpha /
\epsilon))$ rounds for graphs of arboricity bounded by $\alpha$.

As the result of \cite{Solomon18} uses the algorithm of \cite{Bar-YehudaCS17} as a black box, plugging $\Delta = \alpha / \epsilon$, our results can also be used. This means all of results stated in this paper also hold for bounded arboricity graphs setting $\Delta = \alpha / \epsilon$.
%, where $W$ is the maximal weight.
We list the previous results and the results of this paper in Table~\ref{tbl} (Adapted from \cite{AstrandS10}).

\begin{table*}[t!]
	\centering{%\hspace*{-0.5cm}
		\resizebox{1.0 \textwidth}{!}{
			\begin{tabular}{|cclll|}
				\hline
				deterministic& weighted& approximation& time $(W=1)$& algorithm\\
				\hline
				no& yes& 2&$O(\log n)$& \cite{GrandoniKP08}\\
				
				no& yes& 2&$O(\log n)$& \cite{Koufogiannakis:2009:DPA:1582716.1582746}\\
				
				yes& no& 3&$O(\Delta)$& \cite{PolishchukS09}\\
				
				yes& no& 2&$O (\log^{4}n)$& \cite{HanckowiakKP01}\\
				
%				yes& no& 2&$O(\Delta+\log^{*}n)$& \cite{PanconesiR01}\\
				
				yes& no& 2&$O(\Delta^{2})$& \cite{AstrandFPRSU09}\\
				
				yes& yes& $ 2+\epsilon$&$O(\log\epsilon^{-1}\log n)$& \cite{KhullerVY94}\\
				
				yes& yes& $ 2$&$O(\log^2 n)$&\cite{KhullerVY94}\\
				
				yes& yes& $ 2+\epsilon$&$O(\epsilon^{-4}\log\Delta)$& \cite{Hochbaum82, KuhnMW06}\\
				
				yes& yes& 2&$O(\Delta+\log^{*}n)$& \cite{PanconesiR01}\\
				
				%				yes& yes& 2&$O(1)$ if $\Delta\leq 3$& [2]\cite{}\\
				
				yes& yes& 2&$O(\Delta)$& \cite{AstrandS10}\\
				
				yes& yes& 2&$O(1)$ for $\Delta\le 3$& \cite{AstrandFPRSU09}\\
				
				yes& yes& $ 2+\epsilon$&$O(\epsilon^{-1}\log\Delta/\log\log\Delta)$& \cite{Bar-YehudaCS17}\\	
				
				yes& yes& $ 2+\frac{\log\log\Delta}{\log\Delta}$&$O(\log\Delta/\log\log\Delta)$& \cite{Bar-YehudaCS17}\\								
				yes& yes& ${2+\epsilon}$&$O\parentheses{\frac{\log \Dv}{\log\log \Dv} + {\frac{\log\epsilon^{-1}\log \Dv}{\log^2 \log \Dv}}}$& This work\\	
				
				yes& yes& ${{2+\parentheses{\log\Delta}^{-c}}}$&$O(\log\Delta/\log\log\Delta)$& This work, $\forall c=O(1)$\\								
				yes& yes& 2&$O(\log n\log{\Delta}/\log^2\log{\Delta})$& This work\\
				
				\hline
				%				\tabularnewline
				%				\hline
				%				\hline
				%				\multirow{1}{*}{\paramVolEst} & \multirow{1}{*}{$O\parentheses{\log\delta^{-1}\parentheses{\eps^{-2}+\log n}}$} & \multirow{1}{*}{$O(1)$} & \multirow{1}{*}{$O(1)$}\tabularnewline
				%				\hline
				%				\multirow{1}{*}{\paramFreq{}} & \multirow{1}{*}{$O\parentheses{\eps^{-1}\log^2\delta^{-1}\parentheses{\eps^{-2}+\log n}}$} & \multirow{1}{*}{$O(\log\delta^{-1})$} & \multirow{1}{*}{$O(\log\delta^{-1})$}\tabularnewline
				%				\hline
				%				%			\multirow{1}{*}{\paramHHs{}} & \multirow{1}{*}{$O\parentheses{\eps^{-1}\log\delta^{-1}\log|\mathcal U|\log\Psi\parentheses{\eps^{-2}+\log n}}$} & \multirow{1}{*}{$O(\log|\mathcal U|\log\Psi)$} & \multirow{1}{*}{$O(\eps^{-1}\log|\mathcal U|\log\Psi)$}\tabularnewline
				% \multirow{1}{*}{\paramHHs{}} &
				% \multirow{1}{*}{$O\parentheses{\eps^{-1}\log\delta^{-1}\log|\mathcal
				% U|\log\Psi\parentheses{\eps^{-2}+\log
				% n}}$} &
				% \multirow{1}{*}{$O(\log|\mathcal
				% U|\log\Psi)$} &
				% \multirow{1}{*}{$O(\eps^{-1})$}\tabularnewline
				%				\hline
			\end{tabular}
		}
	}
	\caption{
		In the table (adapted from \cite{AstrandS10}), $n=|V|$ and $\epsilon\in(0,1)$. The running times are stated for the case of unit weight vertices. For randomized algorithms the running times hold in expectation or with high~probability.
	}
	\label{tbl}
\end{table*}

\section{The MWVC local ratio template}
\label{sec:LR}
In this section we overview \cite{Bar-YehudaCS17}'s local ratio
paradigm for approximating MWVC.  We note that the template does not
assume anything about the model of computation and that our algorithms
will fit into this framework. This template can also be viewed via the
primal-dual schema.

Let  $G=(V,E)$ denote a graph with a vertex-weight function $w: V
\rightarrow \mathbb{R}^{+}$. An edge-weight function $\delta: E
\rightarrow \mathbb{R}^{+}$ is \emph{$G$-valid} if for every vertex
$v$ the incident edges weight sum does not exceed $w(v)$; that is,
$\delta$ is $G$-valid if $\forall v \in V:\sum_{v \ni e}{\delta(e)}
\leq w(v)$. (In fact, a $G$-valid function $\delta$ is a feasible
solution to the dual edge packing LP.)

Next, for a $G$-valid function $\delta$, define the vertex-weight
function $\tilde{w}_{\delta}: V \rightarrow \mathbb{R}^{+}$ by
$\tilde{w}_{\delta}(v) = \sum_{e: v \in e}{\delta(e)}$.  Let
$S_{\delta}=\{v \in V ~|~ w(v)-\tilde{w}_{\delta}(v) \leq \epsilon'
w(v) \}$ be the set of vertices for which $w$ and $\tilde{w}$ differ
by at most $\epsilon' w(v)$, for $\epsilon'= \epsilon/(2+\epsilon)$.
We refer to vertices in $S_\delta$ as $\eps'$-\emph{tight}
vertices. The essence of the template consists of two parts: 
(1)~The sum of the weights of the vertices in $S_\delta$ is at most
$(2+\eps)$ times the weight of an optimal solution to MWVC.  
(2)~When the algorithm terminates, $S_\delta$ is a vertex cover. 
\begin{theorem}\emph{(\cite{Bar-YehudaCS17})}
  \label{thm:lr}
  Fix $\epsilon>0$ and let $\delta$ be a $G$-valid function. Let $OPT$
  be the sum of weights of vertices in a minimum weight vertex cover
  $S_{OPT}$ of $G$. Then $\sum_{v \in S_{\delta}}{w(v)} \leq
  (2+\epsilon)OPT$. In particular, if $S_{\delta}$ is a vertex cover,
  then it is a $(2+\epsilon)$-approximation for MWVC for $G$.
\end{theorem}

\section{A fast distributed implementation}
\label{sec:dist}
In this section, we present a modification of the distributed
algorithm for MWVC of Bar-Yehuda \etal~\cite{Bar-YehudaCS17}. The pseudo-code for our algorithm is given in Algorithm~\ref{alg}. In this
section we assume that the maximal degree $\Delta$ is known to all
vertices.~\footnote{A polynomial upper bound of $\Delta^{O(1)}$ would
  yield the same asymptotic bound on the number of rounds.}  In
Section~\ref{sec:noDelta} we provide an algorithm with a slightly
higher running time in which this assumption is lifted.  

%Since the case of constant maximum degree admits a $2$-approximation in $O(1)$ rounds (e.g., primal-dual algorithm), 

%
%Our algorithm is deterministic and requires for every vertex $v$ only
%$O(\log{d(v)}/\epsilon'\log\log{d(v)})$ rounds, 
%$O(\log\epsilon^{-1}\log d(v)/\log\log d(v))$ rounds, 
%
%where $d(v)$ is the degree of $v$ in $G$, or $O(\log 1/\epsilon')$ if $d(v)\leq 16$. Here $\epsilon'= \epsilon/(2+\epsilon)$ where $\epsilon=O(1)$, which means that $\epsilon'=\Theta(\epsilon)$.
For clarity of presentation, we first describe an implementation for the LOCAL model. This algorithm can be easily adapted to the CONGEST model using the techniques of \cite{Bar-YehudaCS17}.

\paragraph{Overview of Algorithm~\ref{alg}} 
The algorithm  uses the following parameters:
(i)~$\epsilon' \triangleq \epsilon/(2+\epsilon)$, (ii)~$\gamma\in (0,1)$, (iii)~$z\triangleq \ceil{\log_\gamma\epsilon'}$. 
The parameter $\eps'$ is used for defining tightness of the dual packing constraint.
The parameter $\gamma$ is used for defining levels. Loosely speaking, in every iteration the weight of a vertex is reduced, and the level of a vertex is proportional to $\log_\gamma (w_i(v)/w_0(v))$. 
The parameter $z$ is used to bound the number of levels till a vertex becomes $\eps'$-tight, meaning that $w_i(v)/w_0(v) \leq \eps'$.

Our algorithm, listed as Algorithm~\ref{alg}, is a variation of the Algorithm of Bar-Yehuda
\etal~\cite{Bar-YehudaCS17} with a few modifications.  We begin with a description of the common features.
In the course
of the algorithm, the weight of each vertex is reduced. Once a vertex
$v$ becomes $\epsilon'$-tight (i.e., the reduced weight is an
$\epsilon'$-fraction of its original weight) it decides to join the
vertex cover and terminates after sending the message $(v,cover)$ to its
remaining neighbors. The message $(v,cover)$ causes the neighbors of $v$
to erase $v$ from their list of remaining neighbors. If a vertex $v$
loses all its neighbors (i.e., becomes isolated), it decides that $v$ is not in the vertex cover, and terminates. Upon termination, the
$\epsilon'$-tight vertices constitute a vertex cover.

The handling of offers is as in~\cite{Bar-YehudaCS17}. Vertex $v$
sends an irrevocable offer $\Deal_i(v,u)$ to every $u\in N_i(v)$.
The offers are allocated from the vault. The responses to the offers
are allocated greedily from the bank, namely $v$'s responses satisfy:
$\Budget_i(v,u)\leq \Deal_i(u,v)$ and $\sum_{u} \Budget_i(v,u)\leq
\Bank_i(v)$. The updating of the weights can be interpreted as
follows. For every edge $e=\{u,v\}$ the dual edge packing variable
$\delta(e)$ is increased by $\Budget_i(u,v)+\Budget_i(v,u)$. The
remaining weight satisfies $w_{i+1}(v) =w_0(v)-\sum_{e\ni v}
\delta(e)$. 
Note that each iteration of the while-loop requires a constant number of communication rounds.

The first modification is that we attach a \emph{level} to each vertex as follows. 
Let $w_{i}(v)$ denote the weight of $v$ in the beginning of iteration $i$ of the while-loop.
The level of $v$ in iteration $i$ satisfies $\ell_i(v)= 1+\floor{\log_{\gamma} \frac{w_{i}(v)}{w_0(v)}}$.
Note that the initial level is one, and that if the level of $v$ is greater than $z$, then $v$ is $\eps'$-tight (see Claim~\ref{claim:inv}). 

The second modification is how we partition $w_i(v)$ between the vault and the bank.
Instead of using a fixed fraction of the initial weight for the vault, our vault decreases as the level of the vertex increases.
Formally, $\Vault_i(v)\triangleq w_0(v)\cdot \gamma^{\ell_i(v)}$. The bank is the rest of the weight, namely, 
$\Bank_i(v) \triangleq w_i(v) - \Vault_i(v)$.

The third modification is that in each iteration, every vertex $v$
only sends offers to its remaining neighbors with the smallest
level. Let $N_i(v)$ denote the set of remaining neighbors of $v$ in the beginning of the $i$th iteration. 
The smallest level of the neighbors of $v$ is defined by $\ell'_i(v)\triangleq \min \{\ell_i(u) \mid u\in N_i(v)\}$.
The set of neighbors of lowest level is defined by $N'_i(v) \triangleq \{u\in N_i(v) \mid \ell_i(u)=\ell'_i(v)\}$.
Let $d'_i(v)=|N'_i(v)|$. The size of each offer sent is $\Vault_i(v)/d'_i(v)$.

Note that if $\gamma=\eps'$, then Algorithm~\ref{alg} reduces to the
BCS Algorithm because there is just one level, and the vault size is
fixed and equals to $\eps' \cdot w_0(v)$. On the other hand, if
$\gamma=1/2$, then there are $O(\log 1/\eps)$ levels. Per level
$\ell_i(v)$, the algorithm can be viewed as a version of the BCS algorithm with $\eps'=1/2^{\ell_i(v)}$. This also explains why our algorithm may be adapted to the CONGEST model of distributed computation using the techniques of \cite{Bar-YehudaCS17}. In essence they give an adaptation for a single level of our algorithm, which can easily be extended to multiple levels.

 \RestyleAlgo{boxruled} \LinesNumbered
\begin{algorithm}[htbp]
	\label{alg}
	\caption{A distributed $(2+\epsilon)$-approximation algorithm for MWVC, code for vertex~$v$. (Listing taken from~\cite{Bar-YehudaCS17} and edited to include our modifications.)}
        $\gamma = $ parameter in the interval $(0,1)$.\\
	$\epsilon'= \epsilon/(2+\epsilon)$\\
%	$\gamma= \epsilon'^{1/z}$\\
	$z= \ceil{\log_\gamma\epsilon'}$\\
	$w_0(v) = w(v)$\\
	$\ell_0(v) = 1$\\
	$N_0(v) = N(v), d_i(v) \triangleq |N_i(v)|$\\
	//Let $N'_i(v)$ be the set of neighbors of lowest level in iteration $i$ and $d'_i(v) \triangleq |N'_i(v)|$ \\ 
%	$\Vault_0(v) = w_0(v)\cdot 2^{-\ell_0(v)}$\\
	$i=0$\\
	\While{$true$}
	{
		$\Vault_i(v) = w_0(v)\cdot \gamma^{\ell_i(v)}$\\
		$\Bank_i(v) = w_i(v) - \Vault_i(v)$\\
		$w_{i+1}(v) = w_i(v) $ and $\ell_{i+1}(v) \gets \ell_i(v)$\\
                % guy: moved to end of iteration
		%$N_{i+1}(v) = N_i(v)$\\
		\ForEach{$u \in N'_i(v)$}
		{		
			$\Deal_i(v,u) = \Vault_i(v)/d'_i(v)$\\
			Send $\Deal_i(v,u)$ to $u$\\
			Let $\Budget_i(u,v)$ be the response from $u$\\
			$w_{i+1}(v) = w_{i+1}(v) - \Budget_i(u,v)$\\
			\iffalse
			
			\If{$\Budget_i(u,v) < \Deal_i(v,u)$}
			{
				%$N_{i+1}(v) = N_{i+1}(v)\setminus\{u\}$
				$u$ leveled up! Delete if needed.
			}
		\fi
		}
                % guy: modified because incoming requests need not come from N'_i(v)
		Let $u_1\dots u_{m_i}$ be an arbitrary order of neighbors that sent requests in this iteration\\
		\ForEach{$k=1,\dots,m_i$}
		{
			Let $\Deal_i(u_k,v)$ be received from $u_k \in N'_i(v)$\\
			$\Budget_i(v,u_k) = \min\{\Deal_i(u_k,v), \Bank_i(v)-\sum_{t=1}^{k-1}{\Budget_i(v,u_t)} \}$\\
			Send $\Budget_i(v,u_k)$ to $u_k$
		}
		%guy: don't see where this is used! $\Bank_{i+1}(v) = \Bank_i(v) - \sum_{k=1}^{d_i(v)}{\Budget_i(v,u_k)}$\\
		$w_{i+1}(v) = w_{i+1}(v) - \sum_{k=1}^{d_i(v)}{\Budget_i(v,u_k)}$\\
%		$d_{i+1}(v) = |N_{i+1}(v)|$\\
%  guy: moved to end of loop!
		%$i=i+1$\\
		\If{$w_{i+1}(v) \neq 0$ and $w_{i+1}(v) \leq  \Vault_i(v)$}%  w_0(v) \cdot \gamma^{\ell_i(v)}$}
		{
			$\ell_{i+1}(v) = 1+\floor{\log_{\gamma} \frac{w_{i+1}(v)}{w_0(v)}}$\\
		}		
		\If{$w_{i+1}(v) = 0$ or $\ell_{i+1}(v) \geq z+1$}
		{
            		Send $(v,cover)$ to all neighbors\\
            		Return $\InCover$\\
        		}
        		\ForEach{$(u,cover)$ received from $u \in N_i(v)$}
        		{
            		$N_i(v) = N_{i}(v)\setminus\{u\}$\\
            		%$d_i(v) = d_i(v) - 1$
        		}
        		\If{$d_i(v)=0$}
        		{
            		Return $\NotInCover$
        		}
                        $N_{i+1}(v) = N_i(v)$\\
                        $i=i+1$\\
                      }	
\end{algorithm}

We now state the main theorem of this work.

\begin{theorem}
\label{thm:alg}
 Algorithm~\ref{alg}
(with $\gamma=\frac{1}{\sqrt{\log \Delta}}$ if $\Delta> 16$ and $\gamma=0.5$ otherwise)
%If $\Delta> 16$, then for every $\eps\in (0,1)$, Algorithm~\ref{alg}
%(with $\gamma=1/\sqrt{\log \Delta}$) 
is a deterministic distributed
$(2+\epsilon)$-approximation algorithm for MWVC. The number of rounds required for the algorithm to terminate is
$O\parentheses{\frac{\log \Dv}{\log\log \Dv} +
  {\frac{\log\epsilon^{-1}\log \Dv}{\log^2 \log \Dv}}}$ if $\Delta> 16$ and $O(\log\epsilon^{-1})$ otherwise.
\end{theorem}

\subsection{Proof of Theorem~\ref{thm:alg}}
\paragraph{Notation.} In the analysis we use
$w_i(v),\ell_i(v),d'_i(v)$ to denote the value of these variables
at the beginning of the
$i$th iteration. 
%(We count iterations of the while-loop according to the value of variable $i$.)

% guy
\noindent \medskip The following claim states an invariant that
Algorithm~\ref{alg} satisfies. 
\begin{claim}\label{claim:inv}
  The following invariant holds in every iteration of the while-loop:
  \begin{align}\label{eq:inv}
    \gamma^{\ell_i(v)} < \frac{w_i(v)}{w_0(v)} \leq \gamma^{\ell_i(v)-1}
  \end{align}
  Hence, (i)~$\Vault_i(v)<w_i(v)$ and (ii)~if $\ell_{i+1}(v)\geq z+1$, then $\frac{w_{i+1}(v)}{w_0(v)}\leq \epsilon'$.
\end{claim}

This invariant of Claim~\ref{claim:inv} implies, among other things,
that every vertex that decides to join the vertex cover is
$\eps'$-tight. This property, together with the fact that the set of
vertices that join the vertex cover constitute a vertex cover leads to
the proof that Algorithm~\ref{alg} is a $(2+\eps)$-approximation
algorithm. An analogous lemma and its proof also appears
in~\cite{Bar-YehudaCS17}. We remark that termination of the algorithm
is implied by the upper bound on the number of iterations of the
while-loop proved in the sequel.
\begin{lemma}[{\cite[Lemma 3.2]{Bar-YehudaCS17}}]
  For every $\epsilon,\gamma \in (0,1)$, upon termination Algorithm~\ref{alg} computes a
  $(2+\epsilon)$-approximate solution to MWVC.
\end{lemma}

In the following lemma we show that, for every vertex $v$ and every iteration of the while-loop,  either
many of $v$'s neighbors of the smallest level have increased their level or $v$'s  weight has
decreased significantly.
%It remains to bound the number of rounds. We do so in the following lemma, in which we show that in each iteration either enough weight is reduced or enough neighbors enter the vertex cover.
\begin{lemma}\label{lemma:winwin}
  Let $\Kv>1$. Let $i$ be an iteration of the while-loop in the
  execution of Algorithm~\ref{alg} by vertex $v$ in which $v$ does not
  join the cover.  At least one of following conditions must hold:
\begin{enumerate}%[{Condition} (1):]
%\item $d_{i+1}(v) \leq d_i(v)/\Kv$
\item  At least $d'_i(v)(1-1/\Kv)$ of the neighbors of
  $v$ of the lowest level have increased their level. Formally, If $\ell'_{i+1}(v)=\ell'_i(v)$, then $d'_{i+1}(v) < d'_i(v)/K$.\label{condA}
  %let $\ell'_i(v)\triangleq \min \{\ell_i(u) \mid u\in N_i(v)\}$. 
\item $w_{i+1}(v) \leq w_i(v)-w_0(v)\gamma^{\ell_i(v)}/\Kv$.\label{condB}
\end{enumerate}
\end{lemma}
\begin{proof}
  Assume that $\ell'_{i+1}(v)=\ell'_i(v)$ and
  $d'_{i+1}(v) \geq d'_i(v)/K$.  Note that if the level of a vertex
  remains unchanged (i.e., $\ell'_{i+1}(v)=\ell'_i(v)$), then either
  $w_{i+1}(v)=0$ or $w_{i+1}(v)> \Vault_i(v)$.  If $w_{i+1}(v)=0$,
  then $v$ joins the cover, a contradiction.  If
  $w_{i+1}(v)> \Vault_i(v)$, then the bank was not exhausted and
  $\Budget_i(u,v) = \Deal_i(v,u)$.  To conclude, at least
  $d'_i(v)/\Kv$ vertices $u \in N'_i(v)$ responded with
  $\Budget_i(u,v) = \Deal_i(v,u)$. This implies that
\begin{eqnarray*}
w_i(v) - w_{i+1}(v) &\geq& \frac{d'_i(v)}{\Kv} \cdot\frac{\Vault_i(v)}{d'_i(v)}\\
&=& w_0(v)\gamma^{\ell_i(v)}/\Kv.
\end{eqnarray*}
\end{proof}
\begin{lemma}\label{lemma:iterations}
For every $\gamma\in (0,1)$ and $\Kv>1$, the number of iterations of the while-loop for every vertex $v$ is bounded by:
  \begin{align*}
z\cdot \parentheses{\frac{\Kv}{\gamma} + \frac{\log{d(v)}}{\log{\Kv}}}.
  \end{align*}
\end{lemma}

  \begin{proof}
  The number of levels is bounded by $z$. Hence it suffices to prove
  that the number of rounds per level is at most
  $K/\gamma + \log_K d(v)$.  Indeed, the number of rounds that
  satisfy Condition~\ref{condA} per level is bounded by $\log_K d(v)$ because
  $d'_i(v)$ is divided by at least $K$ in each such iteration.

  We now bound the number of iterations that satisfies 
  Condition~\ref{condB} per level. By Claim~\ref{claim:inv},
  $w_0(v)\cdot \gamma^{\ell_i(v)-1} \geq w_i(v)$. Hence, the number of
  iterations that satisfies Condition~\ref{condB} is bounded by
  $K/\gamma$, as required, and the lemma follows.
\end{proof}

\iffalse
\begin{claim}\label{claim:z}
If $\Delta> 16$ and $\gamma=1/\sqrt{\log \Delta}$, then $z\leq 8$ or $z\leq 8\cdot \log_\gamma \eps$.
\end{claim}
%Guy's lemma
\begin{proof}
  Assume for the sake of contradiction that neither hold. Then
  $z>(8+8\cdot \log_\gamma \eps)/2$.  But $z<1+\log_\gamma
  \eps'=1+\log_\gamma \eps -\log_\gamma (2+\eps)$.  Hence
  $-\log_\gamma (2+\eps)>3+3\cdot \log_\gamma \eps$.  On one hand,
  $-\log_\gamma (2+\eps)>3$ implies $2+\eps>1/\gamma^3> 8$ (Note
  that $\Delta> 16$ implies that $\gamma< 2$.)  On the other
  hand, $-\log_\gamma (2+\eps)>3\cdot \log_\gamma \eps$ implies that
  $2+\eps>1/\eps^3$, which hold only if $\eps<1$, a contradiction, and the claim follows.
\end{proof}
\fi
%\newpage 
We now prove Theorem~\ref{thm:alg}.
\begin{proof}%[Proof of Theorem~\ref{thm:alg}]
First, consider the case where $\Delta \le 16$. We set $\gamma=0.5$ (hence, $z=O(\log\epsilon^{-1})$) and $\Kv=2$. Lemma~\ref{lemma:iterations} immediately shows that the termination time is $O(\log\epsilon^{-1})$.
Next, assume that $\Delta > 16$ (thus hereafter: $\log\log\Delta>
2$, $\log\Delta / \log\log\Delta > 2$, $1/\sqrt{\log \Delta}<1/2$, and
$\sqrt{\log \Delta}/{\log \log \Delta}>1$).	
We set $\gamma = 1/\sqrt{\log \Delta}$ and $\Kv=\sqrt{\log \Dv} / \log \log \Dv$. 
%First note that as $\Dv > 16$ it always holds that $\Kv>1$.  
Now we can express the running time as:
\begin{align*}
z\cdot \parentheses{\frac{\Kv}{\gamma} + \frac{\log{d(v)}}{\log{\Kv}}} \le	z\parentheses{\Kv\gamma^{-1} + \frac{\log{\Dv}}{\log{\Kv}}} = z\parentheses{\frac{\log \Dv}{\log \log \Dv}  + \frac{\log{\Dv}}{0.5\log \log \Dv - \log \log \log \Dv}} = O\parentheses{\frac{z\log \Dv}{\log \log \Dv}}.
\end{align*}

Let us analyze the running time according to the values of
$\epsilon$. First, consider the case where $\epsilon^{-1}=\log^{O(1)}
\Dv$. Since $\eps\in (0,1)$, it follows that $\eps'=\Theta(\eps)$. We get that 
$z < 1+\log_\gamma \eps' = O(1 + \log \epsilon^{-1} / \log \gamma^{-1}) =
O(1+\log \log \Dv / \log \log \Dv ) = O(1)$. Thus, the total running
time for this case is $O(\log \Dv / \log \log \Dv)$.  Next we consider
the complementary case, where $\epsilon^{-1}=\log^{\omega(1)}
\Dv$. This means that $\log \epsilon^{-1} = \omega(\log \log
\Dv)$. Therefore, $z = O(\log \epsilon^{-1} / \log \gamma^{-1}) =
O(\log \epsilon^{-1} / \log \log \Dv)$, and the running time for the
second case is given by $O(\log \epsilon^{-1} \log \Dv / \log^2 \log
\Dv)$. Thus, we may express the running time of our algorithm
asymptotically as:
$$O(\log \Dv / \log \log \Dv + \log \epsilon^{-1} \log \Dv / \log^2 \log \Dv).$$

%Guy's lemma
\iffalse

We bound the number of rounds by setting $K=\sqrt{\log \Delta}/\log\log \Delta$. 
This choice of $K$ implies that $K/\gamma=O(\log \Delta/\log K)$. Hence $(K/\gamma+\log\Delta/\log K)=O(\log \Delta/\log \log \Delta)$.
We now consider two cases:
\begin{enumerate}
\item If $z\leq 8$, by Lemma~\ref{lemma:iterations}, the number of rounds
  is $O(\log \Delta / \log\log \Delta)$ rounds.
\item If $z\leq 8\cdot \log_\gamma \eps$, then 
$z=O(\frac{\log 1/\eps}{\log \log \Delta})$.
By Lemma~\ref{lemma:iterations}, the
  number of rounds is $O(z\cdot \log \Delta / \log\log \Delta)= {\frac{\log1/\epsilon\log \Dv}{\log^2 \log \Dv}}$ rounds.
\end{enumerate}
\fi
The number of rounds is bounded as required, and the theorem follows.
\end{proof}

\section{An algorithm without knowing $\Delta$}\label{sec:noDelta}
The bound on the round complexity in Theorem~\ref{thm:alg} assumes
that every vertex knows the maximum degree $\Delta$ (or a polynomial
upper bound on $\Delta$). This is required in order to determine the value of $\gamma$.  
In this section we consider the setting in which $\Delta$ is unknown to the vertices. 

Note that the analysis in Lemma~\ref{lemma:iterations} is per
a vertex. Hence, in the analysis of the round complexity, we may use a
different $\Kv$ per a vertex. Let $\Kv_v$ denote the value of $\Kv$ that is
used in the analysis for bounding the round complexity of $v$ . 

\medskip\noindent
We propose two alternatives for the setting of unknown maximum degree, as follows:
\begin{enumerate}
\item In the first setting, we simply set $\gamma=1/2$ in the
  algorithm. For the analysis, we consider $\Kv_v= \frac{\log
    d(v)}{\log\log d(v)}$, where $d(v)$ denotes the degree of $v$.
  Plugging in these parameters in Lemma~\ref{lemma:iterations} gives a
  round complexity of $O\parentheses{\frac{\log\epsilon^{-1}\log
    d(v)}{\log\log d(v)}}$.
%\item In the first setting, we simply set $\gamma=1/2$ in the
%algorithm. For the analysis, we consider $\Kv_v= \max\set{\frac{\log
%		d(v)}{\log\log d(v)},2}$, where $d(v)$ denotes the degree of $v$.
%Plugging in these parameters in Lemma~\ref{lemma:iterations} gives a
%round complexity of $O\parentheses{\frac{\log\epsilon^{-1}\log
%		d(v)}{\log\log d(v)}}$ if $d(v) > 16$ and $O(\log\epsilon^{-1})$ otherwise.

\item
 For any $q=O(1)$, we can set
$\gamma=\epsilon^{1/2q}$ (hence, $z=O(1)$).  An analysis with $K_v = \frac{\gamma \log d(v)}{\log\log d(v)}$ shows that $v$ terminates in the optimal $O\parentheses{\log d(v) / \log\log d(v)}$
rounds for any $\epsilon > (\log \Delta)^{-q}$. 
This is because
$$\Kv_v=\frac{\epsilon^{1/2q} \log d(v)}{\log\log d(v)} > \frac{\log^{-0.5} d(v)\log d(v)}{\log\log d(v)} = \frac{\log^{0.5} d(v)}{\log\log d(v)}.$$
That allows us to express the running time as
 $$z\parentheses{\Kv_v\gamma^{-1} +
  \log{d(v)}/\log{\Kv_v}} =O\parentheses{\log d(v) / \log\log d(v)}.$$
%Notice that $\Kv_v$ is an analysis parameter, does not
%affect the algorithm itself, and thus we do not require $v$ to know $\Delta$.
\end{enumerate}

\bibliographystyle{alpha}
\newcommand{\etalchar}[1]{$^{#1}$}

\end{document}